\documentclass[showpacs,amsmath,amssymb,twocolumn,pra,superscriptaddress,notitlepage]{revtex4-2}

\def\preprint{0}

\def\draft{0}
\def\aqis{0}
\def\noappendix{0}
\def\supplematerial{0}
\input{./tex/base.sty}

\begin{document}
\title{
Error Interference in Quantum Simulation}
\date{\today}

\author{Boyang Chen}
\affiliation{Department of Computer Science and Technology, Tsinghua University, Beijing 100084, China}
\author{Jue Xu}
\affiliation{QICI Quantum Information and Computation Initiative, School of Computing and Data Science, The University of Hong Kong, Pokfulam Road, Hong Kong}
\author{Xiao Yuan}
\email[]{xiaoyuan@pku.edu.cn}
\affiliation{Center on Frontiers of Computing Studies, Peking University, Beĳing 100871, China}
\affiliation{School of Computer Science, Peking University, Beĳing 100871, China}
\author{Qi Zhao}
\email[]{zhaoqi@cs.hku.hk}
\affiliation{QICI Quantum Information and Computation Initiative, School of Computing and Data Science, The University of Hong Kong, Pokfulam Road, Hong Kong}

\begin{abstract}

Understanding algorithmic error accumulation in quantum simulation is crucial due to its fundamental significance and practical applications in simulating quantum many-body system dynamics. Conventional theories typically apply the triangle inequality to provide an upper bound for the error. However, these often yield overly conservative and inaccurate estimates as they neglect error interference~---~a phenomenon where errors in different segments can destructively interfere.
Here, we introduce a novel method that directly estimates the long-time algorithmic errors with multiple segments,
thereby establishing a comprehensive framework for characterizing algorithmic error interference. We identify the sufficient and necessary condition for strict error interference and introduce the concept of approximate error interference, which is more broadly applicable to scenarios such as power-law interaction models, the Fermi-Hubbard model, and higher-order Trotter formulas. Our work demonstrates significant improvements over prior ones and opens new avenues for error analysis in quantum simulation, offering potential advancements in both theoretical algorithm design and experimental implementation of Hamiltonian simulation.

\end{abstract}
\maketitle

\section{Introduction}
Simulating quantum dynamics is a central application of quantum computation \cite{feynmanSimulatingPhysicsComputers1982, nielsenQuantumComputationQuantum2010}. 
Efficient quantum algorithms have been proposed for realizing the Schr\"odinger equation of a general quantum system~\cite{ciracGoalsOpportunitiesQuantum2012,georgescuQuantumSimulation2014, daleyPracticalQuantumAdvantage2022}, a notoriously difficult task for classical computers. These algorithms have been applied for various tasks, such as in quantum chemistry~\cite{mcardleQuantumComputationalChemistry2020, PhysRevA.90.022305,babbush2015chemical,babbush2014adiabatic} and quantum field theories~\cite{jordanQuantumAlgorithmsQuantum2012,hamedmoosavianFasterQuantumAlgorithm2018,jordanBQPcompletenessScatteringScalar2018,bauerQuantumSimulationFundamental2023}, and also serve as an indispensable subroutine for other fundamental quantum algorithms, such as quantum phase estimation \cite{kitaevQuantumMeasurementsAbelian1995} and the HHL algorithm for solving linear systems \cite{harrowQuantumAlgorithmSolving2009}.
Since Lloyd's proposal of the first digital quantum simulation algorithm based on the product formula  (PF), also known as the Trotter-Suzuki formula, 
novel techniques such as Linear Combination of Unitaries (LCU) \cite{childsHamiltonianSimulationUsing2012,berryExponentialImprovementPrecision2014,berrySimulatingHamiltonianDynamics2015} and Quantum Signal Processing (QSP) \cite{lowOptimalHamiltonianSimulation2017,lowHamiltonianSimulationQubitization2019,gilyenQuantumSingularValue2019} have been developed. These advancements have led to algorithms with optimal or nearly optimal dependence on several key parameters \cite{berryEfficientQuantumAlgorithms2007, berryHamiltonianSimulationNearly2015,childsNearlyOptimalLattice2019,lowOptimalHamiltonianSimulation2017,lowHamiltonianSimulationQubitization2019,gilyenQuantumSingularValue2019,gilyenQuantumSingularValue2019}.
Nevertheless, variants of the product formula are still promising and popular candidates for near-term quantum devices \cite{kimEvidenceUtilityQuantum2023} or intermediate problems \cite{childsFirstQuantumSimulation2018} due to their simplicity, mild hardware requirements, and decent performance in practice.

For a given Hamiltonian $H = \sum_{l=1}^LH_l$ with $L$ terms of $H_l$, the key idea of the product formula method is to divide the long-time evolution $e^{-\ii Ht}$ into several segments $r$ and approximate the short-time evolution $e^{-\ii  Ht/r}$ using the $p$th-order product formula (PF$p$) denoted $\pf_p(t)$ \cite{suzukiGeneralTheoryFractal1991}. 
For example, the PF1 is defined as $\pf_1(t/r):= e^{-\ii H_1 t/r} e^{-\ii H_2 t/r } \cdots e^{-\ii H_L t/r} $ and the algorithmic (Trotter) error of each segment is given by $\epsilon=\|\pf_1(t/r)-e^{-\ii Ht/r } \|$. 
Understanding the Trotter error is an important problem,  providing the basis for analyzing and optimizing the performance of the algorithm.  
A tighter error analysis can help save gate costs in both Hamiltonian simulation \cite{childsFirstQuantumSimulation2018,childsNearlyOptimalLattice2019,childsTheoryTrotterError2021} and other product-formula-related tasks, such as imaginary time evolution \cite{motta2020determining}, quantum Monte Carlo \cite{bravyi2015monte,bravyi2017poly}, quantum adiabatic algorithms \cite{kovalskySelfhealingTrotterError2023}, and quantum phase estimation \cite{yiSpectralAnalysisProduct2021}. 
Many efforts \cite{suNearlyTightTrotterization2021,sahinogluHamiltonianSimulationLowenergy2021,zhaoHamiltonianSimulationRandom2021,zhaoExploitingAnticommutationHamiltonian2021,heylQuantumLocalizationBounds2019,tranDestructiveErrorInterference2020,tranFasterDigitalQuantum2021} have been devoted to have tighter upper bounds of $\epsilon$. However, the total Trotter error estimation is occasionally loose, as it is crudely upper-bounded by $r\epsilon$ through the triangle inequality, which overlooks destructive error interference between different segments.

Interestingly, it has been recently found that Trotter errors from different segments do destructively interfere in the PF1. Thus, the total error may not increase linearly with the number of segments $r$ \cite{tranDestructiveErrorInterference2020}. This error interference phenomenon can be explained from a PF2 perspective \cite{laydenFirstOrderTrotterError2022} and find various applications, e.g., quantum adiabatic algorithms \cite{kovalskySelfhealingTrotterError2023} and quantum phase estimation \cite{yiSpectralAnalysisProduct2021}. However, this interesting phenomenon has only been investigated in a highly specific case, i.e., the two-term Hamiltonian $H=A+B$ with PF1, and its generalization and systematic understanding remain limited. 
In reality, error interference is more prevalent, even in simple cases such as simulating power-law decaying interaction Hamiltonians with PF1 \cite{zhao2022random}.

In this work, we establish a general framework for analyzing algorithmic Trotter error interference in quantum simulation. The framework is based on a novel method that directly estimates the algorithmic Trotter error for a long time evolution, allowing us to derive the necessary and sufficient condition for strict error interference. 
Next, we give a general theoretical lower bound for algorithmic error accumulation, which in various cases excludes the possibility of error interference, e.g., 1D nearest-neighbor Heisenberg model with PF2, and generic Hamiltonians with polynomial series expansion algorithms. 
As applications, we provide more examples with (approximate) error interference phenomena beyond the $H=A+B$ case, including $H=A+B+C$ cases in the Heisenberg model, Fermi-Hubbard model, and power-law interaction models. 
Our results yield significantly improved error bounds compared to the triangle bounds in these cases.
Interestingly, approximate error interference also exists in higher-order product formula approximation of perturbative evolution and can lead to speed-ups in some regions when implementing digital adiabatic algorithms. 
Our results are fundamentally interesting~---~furthering our understanding of the product formula theory, and practically useful~---~directly applicable to the implementation of Hamiltonian simulation using realist quantum hardware. Leveraging our error interference bounds, our results may also inspire better algorithm designs to reduce the total error~\cite{tranFasterDigitalQuantum2021}. 

\ifnum\preprint=1
    \input{theory}
\else
    \section{Interference theory}

\subsection{Trotter error analysis based on the triangle inequality}
We first review the class of Trotter error analysis methods based on the triangle inequality. 
Consider Hamiltonian $H = \sum_{l=1}^LH_l$, then the first-order product formula (PF1), also known as the Trotter-Suzuki formula \cite{suzukiGeneralTheoryFractal1991, lloydUniversalQuantumSimulators1996}, can approximate the quantum dynamics $e^{-\ii H \delta t }$ with the product of the dynamics of each term, that is
\begin{equation}
    \pf_1(\delta t):= e^{-\ii H_1 \delta t} e^{-\ii H_2 \delta t} \cdots e^{-\ii H_L\delta  t} 
    = \prod^{\longrightarrow}_l e^{-\ii H_l\delta  t}. 
    \label{eq:pf1}
\end{equation}
As the Hamiltonian terms $H_l$ may not commute with each other in general, the introduced approximation error (we will call it the Trotter error throughout the paper) is upper bounded as 
$\norm{\pf_1(\delta t) - e^{-\ii H \delta t}} =\mathcal O(\alpha_{\cmm} \delta t ^2)$ \cite{lloydUniversalQuantumSimulators1996}, 
where $\norm{\cdot}$ is the spectral norm and $\alpha_{\cmm}=\sum_{l_1, l_2=1}^L \norm{\qty[H_{l_2}, H_{l_1}]}$. 
One can suppress the Trotter error using a higher-order product formula. 
The second-order product formula can be defined as $\pf_2(\delta t):= \prod^{\rightarrow}_l e^{-\ii H_l\delta t/2} \prod^{\leftarrow}_l e^{-\ii H_l\delta t/2}$,
where $\prod_l^{\leftarrow}$ denotes a product in the reverse order. More generally, a $p$th-order Suzuki product formula (PF$p$), denoted $\pf_p(\delta t)$ \cite{suzukiGeneralTheoryFractal1991} can be defined recursively. 

For a long-time evolution $e^{-\ii H t}$, we can generally divide the whole evolution into $r$ segments and apply PF with time duration $\delta t=t/r$ for $r$ times, i.e., $\pf_1^r(t/r)$. 
In conventional error analyses, the triangle inequality is applied and the total error is therefore bounded as
\begin{equation}
    \norm{\pf_p^r(t/r) - e^{-\ii tH }} \le r  \norm{\pf_p(t/r) - e^{-\ii H t/r }}.
    \label{eq:triangle_bound}
\end{equation}
Numerous studies have concentrated on obtaining more precise bounds for the Trotter error in each segment; however, the overall error cannot be considered stringent in general, as long as the triangle inequality is employed.
In particular, errors in different segments may interfere with each other 
and the triangle-inequality error bound becomes loose. For example, Trotter error interference has been theoretically proved in \cite{tranDestructiveErrorInterference2020,laydenFirstOrderTrotterError2022} and numerically found for Hamiltonians with a power-law rapidly decaying interactions \cite{zhaoHamiltonianSimulationRandom2021}. Yet, the theories in \cite{tranDestructiveErrorInterference2020,laydenFirstOrderTrotterError2022} only work for the special case of PF1 for $H=A+B$. It remains an open problem to quantify Trotter error interference for general Hamiltonians.

\subsection{Framework of Trotter error interference}\label{sec:framework-interference}
Now, we introduce our framework for Trotter error interference. 
We first give its informal definition as follows. 
\begin{definition}[Error interference, informal]\label{def:interference}
For a $p$th-order product formula $\pf_p(t/r)$ approximating $e^{-\ii Ht/r }$, if the total error grows sublinearly to the sum of error of each Trotter, i.e.,
\begin{equation}
    \norm{\pf_p^r(t/r) - e^{-\ii Ht}} = o\qty( r\norm{\pf_p(t/r) - e^{-\ii  H(t/r)}} ),
\end{equation}
where $\norm{\cdot}$ denotes the spectral norm and 
when $r$ is large, then we say that the Trotter error \emph{interferes}.
\end{definition}

To analyze Trotter error interference, we introduce a novel method that directly estimates the Trotter error for a long time evolution. 
We first express the PF evolution by an effective Hamiltonian $H_{\eff}$, i.e.,
\begin{equation}
    e^{-\ii H_{\eff}\delta t} = \pf_p(\delta t) = e^{-\ii(H+H_{\err})\delta t},
\end{equation}
where $H_{\err}$ represents the coherent error of the Hamiltonian with PF. 
For a PF$p$, 
the error term $H_{\err}$ is $\mathcal O(\delta t^{p})$ when $H$ is fixed. 
We can rewrite this error as $H_{\err}=R\delta t^{p}+ R_{\re}\dt^{p+1}$,
where $R$ is the leading-order terms in Trotter error. The high-order remainder, $R_{\re}$, possesses an upper bound independent of $\dt$. 
Compared to conventional ways that estimate the difference between the evolution operators $\pf_p(\delta t)$ and $e^{-\ii H\delta t}$, our new method focuses on the coherent error $H_{\err}$ and its accumulation over time, which can be more naturally accounted for.   

For clarity, we will calculate $R$ and $R_{\re}$ exactly for the first order product formula with two terms. As for PF1 with two terms $\pf_1(\dt) = e^{-\ii H_1 \dt} e^{-\ii H_2 \dt}$, the leading error and the higher-order remaining error can be expressed as 
\begin{align*}
    R &= \frac 12 [H_1, H_2],\\
    \norm{R_{\re}} &= \bigO \qty(\norm{[H_1,[H_1,H_2]]}+\norm{[H_2,[H_2,H_1]]}).
\end{align*} 

We use a perturbative approach to account for the effect of $H_{\err}$ and error accumulation. We consider the spectral decomposition $H = P \Lambda P^\dag$ with $P \in \mathrm{SU}(d)$ and diagonal $\Lambda$. Here $d=2^n$ is the dimension of the Hamiltonian. 
The diagonal elements of $\Lambda$ are the eigenvalues of $H$, and the columns of $P$ are the corresponding eigenstates. Thus, the ideal evolution can be rewritten as $e^{-\ii H \dt}= P e^{- \ii \Lambda \dt} P^\dag$. 
Now, the approximated evolution $e^{-\ii H_{\eff} \dt}$ has a similar spectral decomposition
\begin{equation}
    e^{-\ii H_{\eff}\dt} = (P+\delta P) e^{- \ii (\Lambda + \delta \Lambda) \dt} (P+ \delta P)^\dagger,
\end{equation}
where $\delta P$ and $\delta \Lambda$ are deviations of eigenvectors and eigenvalues, respectively. 
 In general, we derive that $\delta P$ and $\delta \Lambda$ are both $\mathcal O(\|R\| \delta t^{p})$ and the Trotter error in one segment is $\mathcal O(\delta t \delta \Lambda+ \delta P)=\mathcal O(\|R\| \delta t^{p+1})$.
If we apply the conventional approach using triangle inequality, we would have an upper bound of Trotter error $\mathcal O(r(\delta t\delta \Lambda+ \delta P))$, where $\delta P$ and $\delta \Lambda$ both accumulate with the number of segments. However, we can easily check that the effects of $\delta P$ and $\delta t\delta \Lambda$ should not be treated equally. 
In particular, for a long time evolution, 
the approximated evolution is 
$\exp(-\ii r\delta t H_{\eff}) = (P+\delta P) \exp (- \ii r\delta t(\Lambda + \delta \Lambda) ) (P+ \delta P)^\dagger $, where we have used the unitarity property of $P+\delta P$. Thus, the error in eigenvectors 
remains $\delta P$, regardless of the number of segments $r$, while the error in eigenvalues will accumulate as 
$r\delta t\delta \Lambda$. As a result, the total error scales as $\bigO(\delta P+ r \delta t \delta \Lambda)$. Since $\delta P$ is independent of the total time $t$ and the number of segments, the accumulation of error primarily arises from the second term, $ \delta t\delta \Lambda$. 
We show that when the orthogonality condition is satisfied, $\delta \Lambda$ can be much smaller than expected and hence the total error will show a sublinear accumulation, which is the interference phenomenon for the total error. See \cref{fig:schematic}  for the conceptual diagram of an extreme interference case.  

\begin{figure}
    \centering
    \includegraphics[width=0.9\linewidth]{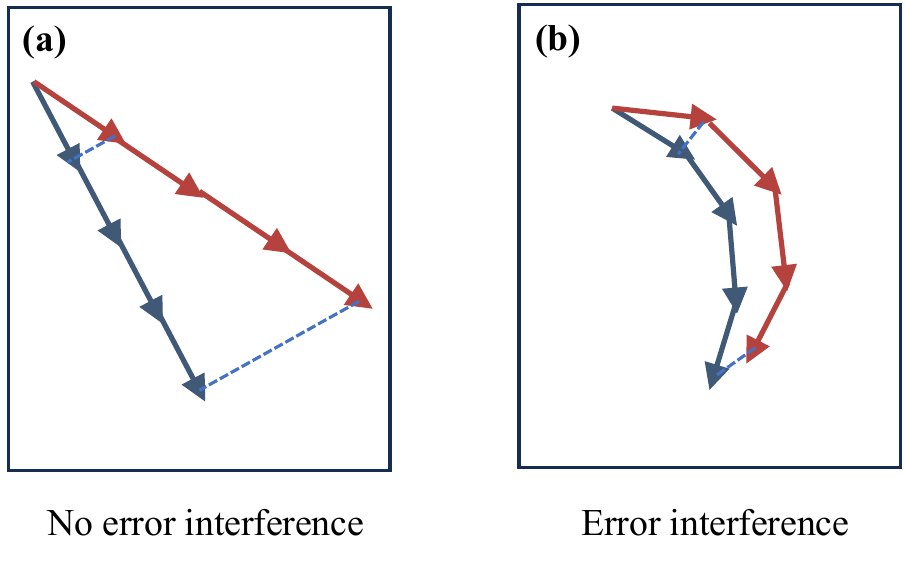}
    \caption{The total error of the simulation can be illustrated by a walk on the plane (a) When the error does not interfere, i.e., $\delta \Lambda$ is the major error term, then the error accumulates linearly as the trotter step grows, the triangle inequality gives a tight bound. (b) When the error interferes, i.e., $\delta P$ is the major error term, then only the first error contributes to the error, the error accumulates sublinearly, the triangle inequality gives a loose bound.}
    \label{fig:schematic}
\end{figure}

\begin{definition}[Orthogonality condition]\label{def:orthogonality}
Consider a given Hamiltonian evolution $e^{-\ii H\dt}$ and its approximation $e^{-\ii H_{\eff}\dt}$.
If the leading-order term $R$ of $H_{\err}=H_{\eff}-H$ satisfies
\begin{align}\label{eq:orthogonality}
    \braket{\psi_i|R|\psi_i}=0
\end{align}
for all eigenstates $\ket{\psi_i}$ of $H$. Then
we say this approximation $e^{-\ii H_{\eff}\dt}$ satisfies an orthogonality condition.
\end{definition}

\begin{theorem}[Necessary and sufficient condition for error interference]\label{thm:fast-convergence}
    The orthogonality condition is a necessary and sufficient condition for Trotter error interference. 
    A PF$p$ evolution can be expressed as $\pf_p(\dt) = \exp(-\ii H\dt -\ii R \dt^{p+1} -\ii R_{\re}(\dt)\dt^{p+2})$, 
    where $R$ is independent of $\dt$ and $R_{\re}(\dt)$ has bounded norm when $\dt$ varies. 
    If $R$ satisfies the \crtlnameref{def:orthogonality}, then the total error can be bounded by
    \begin{align}\label{eq:trotter-interference-approx}
         \norm{\pf_p(t/r)^r - e^{-\ii Ht}} = \bigO \qty(\norm{R}\frac {t^p}{r^p} + \norm{R_{\re}} \frac{t^{p+2}}{r^{p+1}}).
    \end{align} 
    If $R$ does not satisfy the orthogonality condition, then the error term is
    \begin{align}
      \norm{\pf_p(t/r)^r - e^{-\ii Ht}} =   \Omega\left(\max_i \big|\braket{\psi_i|R|\psi_i}\big|\frac{t^{p+1}}{r^p}\right).
    \end{align}
\end{theorem}
The proof of the theorem can be found in~\cref{apd:sec:appendix-interference-bound}.
The orthogonality requirement $\braket{\psi_i|R|\psi_i} = 0$ is equivalent to that there exsits a matrix $M$ such that $[H, M] = R$. 
The condition also has a necessary criterion $\Tr (R H^k) = 0, \forall k \ge 1$.The results in \cite{tranDestructiveErrorInterference2020} and \cite{laydenFirstOrderTrotterError2022} can be recovered from this theorem. For a two-term Hamiltonian $H=H_1+H_2$, the leading-order error term in PF1, i.e. $R = [H_1, H_2] = [H_1+H_2, H_2]$, satisfies the orthogonality requirement. 
Consequently, we conclude that the PF1 error interferes and grows sublinearly.

\cref{thm:fast-convergence} is an asymptotic result. In \cref{thm:tight-general-bound} and \cref{thm:exact-pf1-interference} in Methods, we give tighter upper bounds with concrete prefactors. As a corollary of \cref{thm:fast-convergence}, we get a sufficient condition where the error would accumulate linearly.
\begin{corollary}[Error lower bound]\label{cor:lower-bound}
    If the leading-order error term of $\pf_p(t)$, $R$, satisfies  $\Tr(RH^k) \neq 0$ for some $k \in \mathbb{N}$, then $\norm{\pf_p^r(t/r) - e^{-\ii Ht}} = \Omega(\frac{t^{p+1}}{r^p})$ with a sufficient large  $r$.
\end{corollary}
\noindent As a direct application, we find that for many physical dynamics, the error of PF2 does not interfere. For instance, consider a one-dimensional Ising model, $H = \sum_{j=1}^{n-1}X_jX_{j+1} + \sum_{j=1}^nZ_j$ where $H \coloneqq A+B$, $A$ and $B$ are $X$ and $Z$ terms correspondingly.
It can be confirmed easily that $\Tr (HR) \neq 0$, where $R=\frac{\delta t^3}{12} [A-B, [A, B]]$ is the leading-order error term, thus the error does not interfere. A more detailed discussion can be found in \cref{sec:non-interference-pf2}.

\cref{thm:fast-convergence} can be generalized to any Hamiltonian simulation algorithm; however, it is primarily beneficial for product formulas. The error typically accumulates linearly for other types of simulation algorithms like LCU or QSP. This is due to the equivalent error in LCU or QSP being a function of the Hamiltonian $H$, which does not fulfill the orthogonality condition.

\begin{figure*}[!t]
    \centering
    \sidesubfloat[]{\includegraphics[width=.45\linewidth]{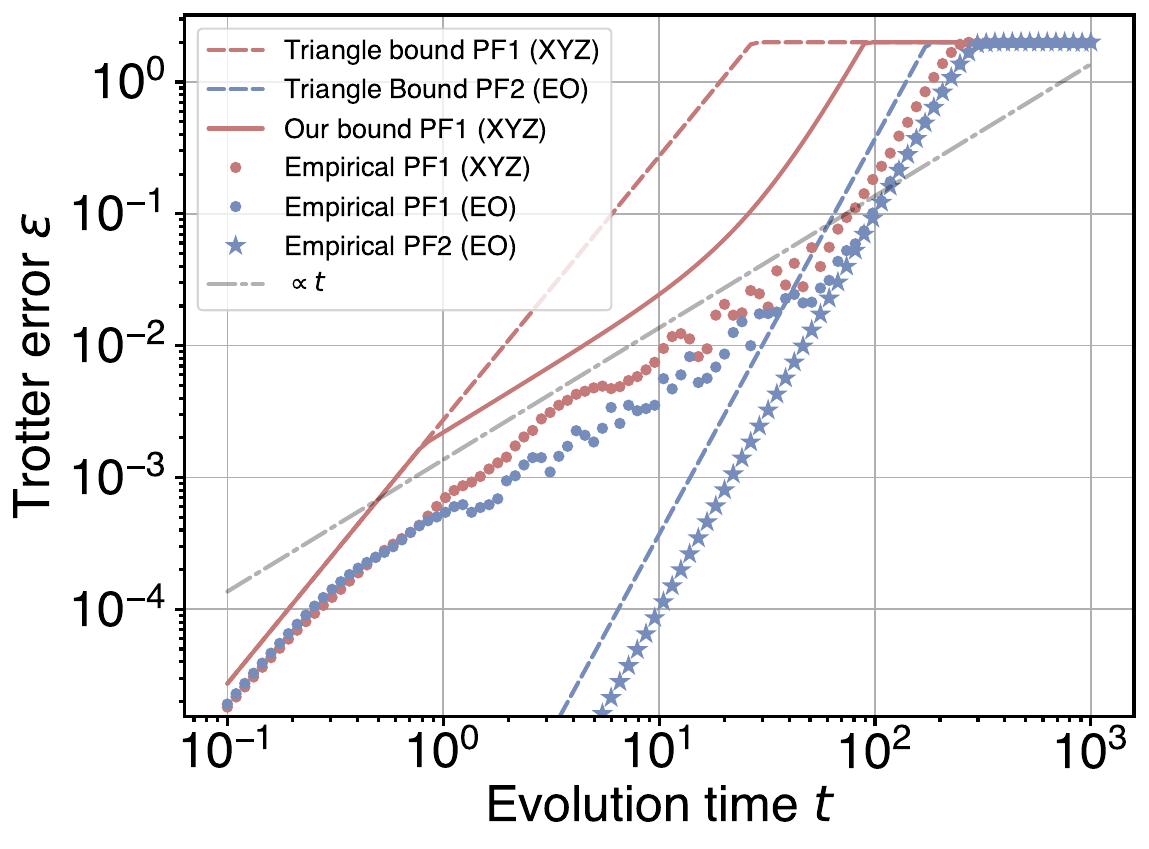}}
    \hfill
    \sidesubfloat[]{\includegraphics[width=.45\linewidth]{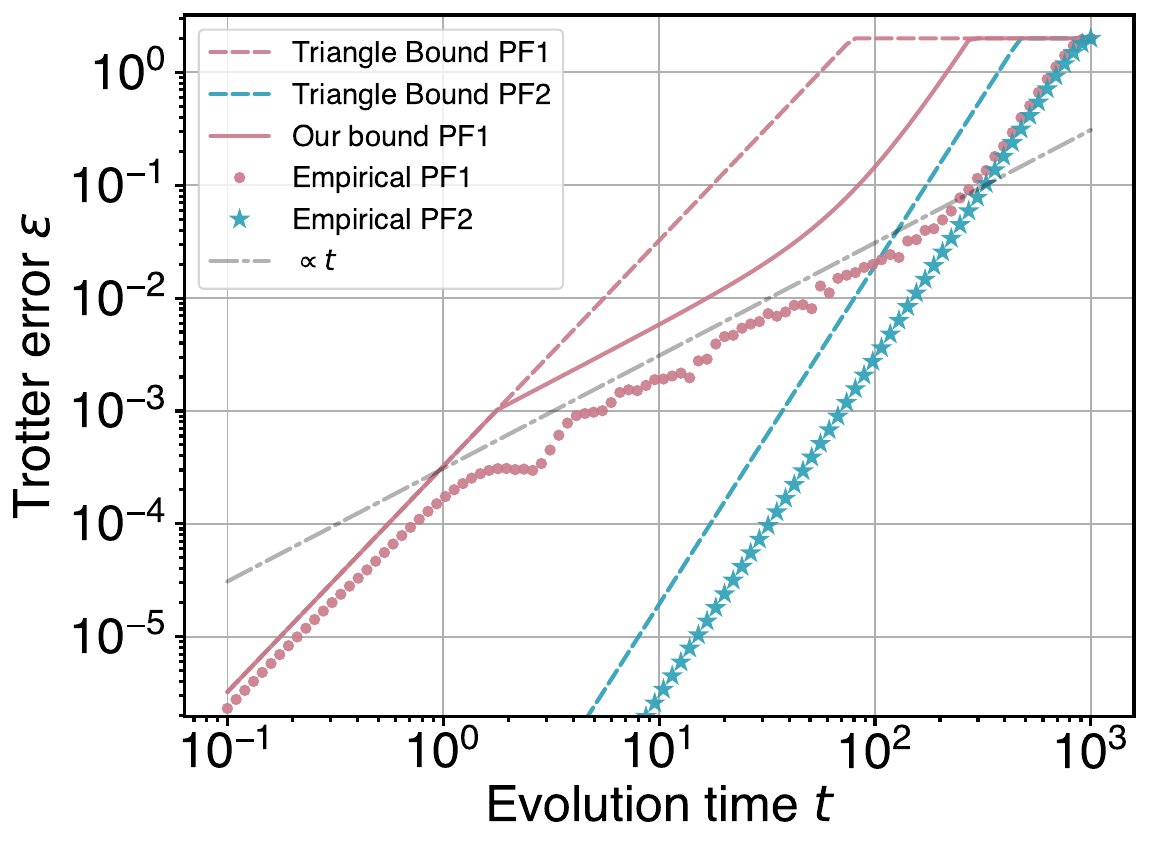}}
    \caption{
    Trotter error VS evolution time $t$ with fixed Trotter step $r=10000$.
    (a) The 1D nearest-neighbor Heisenberg model \cref{eq:heisenberg} of 8 qubits with parameters $J_x=J_y=J_z=2$, $h=0.5$. 
    The red dots are the empirical PF1 Trotter error with the XYZ-group and the red solid line is our interference bound.
    The grey dashed line is plotted proportionally to $t$ as a guide for the eye.
    The empirical error and the triangle bound of PF2 with the even-odd EO-group are also plotted for reference. 
    (b) The 1D Fermi-Hubbard Hamiltonian \cref{eq:hubbard} naturally partitioned into three groups \cref{eq:hubbard_group}.
    We choose the Hamiltonian of 4 sites (represented by 8 qubits)
    and set the interaction parameters as $u=-v=1$. 
    The PF1 empirical error and bounds are plotted in red similar to (a), 
    while the PF2 empirical error and the triangle bound are plotted for comparison.
    }
    \label{fig:tri_group}
\end{figure*}

\cref{cor:lower-bound} can be employed to demonstrate the tightness of upper bounds based on the triangle inequality. If the error is found to be non-interfering, the total error can be estimated using $r\cdot \epsilon$, where $\epsilon$ represents the error in one Trotter step. This approach can significantly reduce the computational resources needed to estimate the minimum required Trotter steps.

Certifying whether the error of evolution satisfies the \crtlnameref{def:orthogonality} is generally challenging. 
However, computational resources can be conserved if the focus is solely on the error of a specific input state, such as eigenstates. 
The subsequent corollary addresses the interference of eigenstates' evolution, expanding upon the results found in \cite{yiSpectralAnalysisProduct2021}.
 
\begin{corollary}
\label{thm:interference-state}
  Consider a $p$th-order product formula $\pf_p(\dt) = \exp(-\ii H\dt -\ii R \dt^{p+1} -\ii R_{\re}(\dt)\dt^{p+2})$. If an eigenstate of $H$, $\ket{\psi}$, satisfies $\braket{\psi|R|\psi} = 0$, then the total error of the evolution of $\ket{\psi}$ with $\pf_p(t/r)^r$ is
    \begin{equation}
        \norm{(\pf_p(t/r)^r-e^{-\ii Ht})\ket{\psi}} = \bigO \left( \norm{R}\frac{t^p}{r^p} + \norm{R_{\re}}\frac{t^{p+2}}{r^{p+1}} \right).
    \end{equation}    
\end{corollary}
\noindent We refer to~\cref{apd:sec:pf-interference-state}. Similar to \cite{yiSpectralAnalysisProduct2021}, the evolution of eigenstates can lead to speed-ups in quantum phase estimation. 
Interference phenomena also occur when considering the long-time average of observables. This corollary characterizes such error interference and shares similar technique with \cite{heylQuantumLocalizationBounds2019} to characterize simulation error with Fourier transform of operators. 
\begin{corollary}
     Consider a $p$th-order product formula $\pf_p(\dt) = \exp(-\ii H\dt -\ii R \dt^{p+1} -\ii R_{\re}(\dt)\dt^{p+2})$.  
     Let $\ket{\psi}$ be any state. If for any $n$
    \begin{equation}
        \lim_{T \to \infty} \frac 1T \int_0^T \braket{\psi|e^{\ii Ht} RH^n e^{-\ii Ht}|\psi} \dd t = 0,
      \end{equation}
    then the total error of the evolution of $\ket{\psi}$ under the dynamics $\pf_p(t/r)^r$ is
    \begin{equation}
        \norm{(\pf_p(t/r)^r-e^{-\ii Ht})\ket{\psi}} = \bigO \left( \norm{R}\frac{t^p}{r^p} + \norm{R_{\re}}\frac{t^{p+2}}{r^{p+1}} \right).
    \end{equation}
\end{corollary}
\noindent The proof of the corollary can be found in~\cref{apd:sec:pf-interference-state}.

In \cref{thm:fast-convergence}, we present the criteria for exact destructive error interference. However, in certain cases, the error may not exhibit complete destructive interference; instead, only a significant portion of the error may interfere. In the following part, we develop a theory to address this approximate interference of Trotter errors.

\ifnum\aqis=0
Suppose that $H$ includes an error term $R$ and that both $H$ and $R$ fail to meet the condition specified in \cref{thm:fast-convergence}. However, if a substantial proportion of $R$ and $H$ adhere to the criterion in \cref{thm:fast-convergence}, then the error will accumulate at a slower pace than what is predicted by the triangle inequality. This idea is formalized in the following theorem.
\begin{theorem}\label{thm:approx-interference-bound}(Approximate error interference bound)
    Consider a $p$th-order product formula $\pf_p(\dt) = \exp(-\ii H\dt -\ii R \dt^{p+1} -\ii R_{\re}(\dt)\dt^{p+2})$. 
    If the leading-order error term $R=R_1+R_2$, and $R_1$ satisfies the \crtlnameref{def:orthogonality}, then the total error $\norm{\pf_p(t/r)^r - e^{-\ii Ht}}$ can be bounded by
    \begin{align}
        \bigO\left( \norm{R_2}\frac{t^{p+1}}{r^p} +
        \norm{R_1}\frac{t^p}{r^p} + \norm{R_{\re}} \frac{t^{p+2}}{r^{p+1}} \right).
    \end{align}
\end{theorem}
\noindent The proof of this theorem can be found in~\cref{apd:sec:proof_approximate_interference}.
This result leads us to conclude that for dynamics of any Hamiltonian $H=H_1+H_2$, if the error of the product formula simulating $H_1$ causes interference and $\norm{H_2}$ is relatively small, then the product formula simulating $H$ will also exhibit a comparably small error. This can explain the error interference of power-law rapidly decaying interaction Hamiltonians and biased Hamiltonians with high-order PF methods. 
We present more details of these examples in \cref{sec:application}.

\section{Applications}\label{sec:application}
\subsection{Heisenberg model and Fermi-Hubbard model with tri-group}

The Heisenberg model is a typical quantum many-body lattice model. In previous works~\cite{tran2020destructive,layden2022first}, error interference is only revealed in a bi-group case $H=H_{\even}+H_{\odd}$. 
Here we find that an XYZ tri-group also exhibits the feature of interference. 
In \cref{fig:tri_group} (a), we compare the empirical error and theoretical error bounds of the Heisenberg model with the nearest-neighbor interaction, that is
\begin{equation}
    H_{\nn}
    =H_X+H_Y+H_Z ,
    \label{eq:heisenberg}
\end{equation}
where $H_X=J_x\sum_{j=1}^{n} X_jX_{j+1}$, $H_y=J_y\sum_{j=1}^{n} Y_jY_{j+1}$, and $ H_Z=\sum_{j=1}^{n} \qty(J_z Z_jZ_{j+1}+hZ_j)$ with the periodic boundary condition.
The PF1 with such XYZ tri-group is implemented as
\begin{equation}
    \pf_1(\dt)
    = e^{-\ii H_X\dt}e^{-\ii H_Y\dt}e^{-\ii H_Z\dt}. 
\end{equation}
The interference bound is obtained by the estimation as in \cref{eq:tight-interference-pf1}. 
We can find that our bound matches the empirical error significantly better than the triangle bound.

Another example with a natural tri-group is the Fermi-Hubbard (FH) model \cite{hubbardElectronCorrelationsNarrow1963}, which can describe a wide range of strongly correlated electron physics, including metal-to-insulator transition, quantum magnetism, and unconventional superconductivity \cite{arovasHubbardModel2022,xuCoexistenceSuperconductivityPartially2024}.
Due to its rich properties such as exotic phase transitions, it is an excellent test for quantum simulation \cite{hensgensQuantumSimulationFermi2017, yangObservationGaugeInvariance2020, campbellEarlyFaulttolerantSimulations2022, shaoAntiferromagneticPhaseTransition2024}.
The FH Hamiltonian takes the form
\begin{equation}
    H_{\textup{FH}}
    =v\sum_{j=1}^{L} \sum_{\sigma\in\qty{\uparrow,\downarrow}} 
    a_{j,\sigma}^\dagger a_{j+1,\sigma} + a_{j,\sigma}^\dagger a_{j+1,\sigma}
    + u \sum_j^L n_{j,\uparrow} n_{j,\downarrow}
    \label{eq:hubbard}
\end{equation}
consisting of the nearest-neighbour hopping term and the on-site interaction between electrons with opposite spins.
The FH model on one-dimensional lattice of $L$ sites can be grouped into three terms \cite{schubertTrotterErrorCommutator2023} i.e.,
$H_{\textup{FH}} = H_{\even} + H_{\odd} + H_{\intt}$ which are
\begin{equation}
\begin{aligned}
    H_{\textup{even}} &= v \sum_{j=1}^{\lfloor{L/2\rfloor}} \sum_{\sigma\in\qty{\uparrow,\downarrow}}  a_{2j-1,\sigma}^\dagger a_{2j,\sigma} + a_{2j,\sigma}^\dagger a_{2j-1,\sigma}, \\
    H_{\textup{odd}} &= v \sum_{j=1}^{\lceil{L/2\rceil}-1} \sum_{\sigma\in\qty{\uparrow,\downarrow}} a_{2j,\sigma}^\dagger a_{2j+1,\sigma} + a_{2j+1,\sigma}^\dagger a_{2j,\sigma},\\ 
    H_{\textup{int}} &= u \sum_j^L n_{j,\uparrow} n_{j,\downarrow},
    \label{eq:hubbard_group}
\end{aligned}
\end{equation}
where $j$ refers to neighboring lattice sites in the first sum, $v$ is the kinetic hopping coefficient, and $u > 0$ the on-site interaction strength. Here,
$a_{j,\sigma}^\dagger$, $a_{j,\sigma}$ and $n_{j,\sigma}=a_{j,\sigma}^\dagger a_{j,\sigma}$ are the fermionic creation, annihilation and number operators, respectively, acting on the site $j$ and spin $\sigma \in \qty{ \uparrow, \downarrow }$.
Similar to the 1D nearest-neighbor Heisenberg model \cref{fig:tri_group} (a), the Trotter error interference can be observed in \cref{fig:tri_group} (b) for this case, and our bound is much better than the triangle bound.

\subsection{Power-law interactions}
We can apply \cref{thm:approx-interference-bound} to the lattice Hamiltonians with power-law interactions. 
More specifically, consider a Heisenberg model with power-law interactions
\begin{equation}
    H_{\pow_\alpha} 
     = \sum_{j=1}^{n-1} \sum_{k=j+1}^n \frac{1}{\abs{j-k}^\alpha}(X_jX_k+Y_jY_k+Z_jZ_k), 
    \label{eq:pow_law}
\end{equation}
where $\alpha>0$ is the decaying coefficient.
This Hamiltonian can be divided into two parts: $H_{\pow_\alpha} = H_{\nn} + H_{\li} $, 
where $H_{\nn}$ refers to the terms acted on a single site or nearest neighbors, and $H_{\li}$ refers to the long-range interactions. 
Then $H_{\li}$ is much smaller than $H_{\nn}$ when $\alpha$ is large. 
The dominating part $H_{\nn}$ satisfies the interference condition. 
According to 
\cref{thm:approx-interference-bound}, the error interferes approximately for a large $\alpha$. 
Our theoretical bound matches the empirical error very well when $\alpha$ is large, supporting the result in \cref{thm:approx-interference-bound}. 

Our analysis also gives a better bound on the required number of Trotter steps, therefore a better resource estimation on the number of total gates. In \cref{fig:power_law} (a), we show the numerical estimation of required Trotter steps $r$, both empirically and theoretically. The empirical error of the simulation grows as $\bigO(n^{2.58})$. The bound by triangle inequalities indicate the number of required $r$ as $\bigO(n^{3.17})$. Our bound could reduce the required Trotter step to $\bigO(n^{2.78})$, which is a better estimation of $r$. Our bound provides approximately 20 times better estimation of the number of Trotter steps over the triangle bound when simulating $50$ qubits, according to the extrapolation of numerical simulation on systems of small size. The interference bound here is the bound from~\cref{thm:exact-pf1-interference}.

\begin{figure}[t]
    \centering
    \includegraphics[width=.95\linewidth]{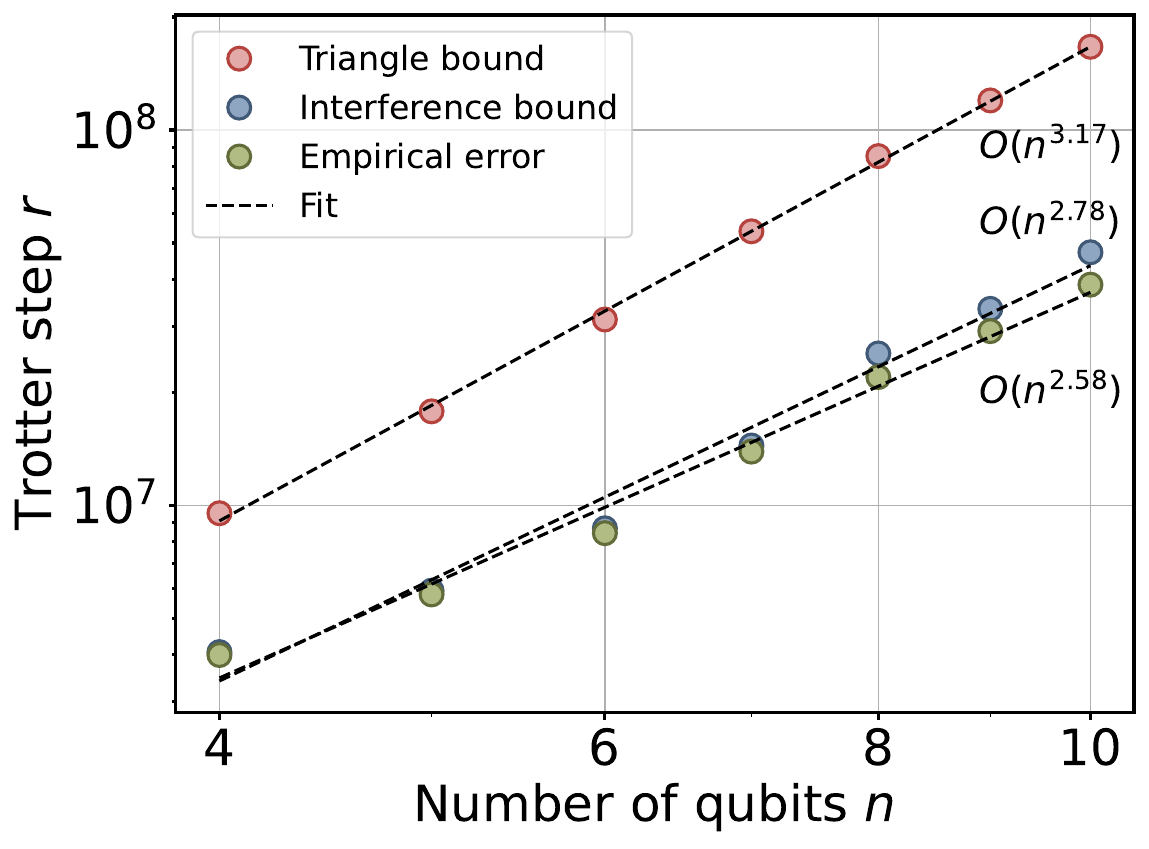}
    \caption{Numerical evaluation of a number of Trotter steps $r$ needed to achieve the fixed error threshold $\epsilon=0.01$ according to different bounds. The simulated quantum system is the 1D lattice power-law Heisenberg model without external magnetic field and the decaying coefficient $\alpha = 4$. The simulation time is 30 times the system size. 
    The empirical curve suggests an $\bigO(n^{2.58})$ scaling of Trotter steps, which is close to our theoretical interference bound $\bigO(n^{2.78})$. 
While the estimate derived from the triangle inequality suggests that the required number of Trotter steps is $\bigO(n^{3.17})$.
    }
    \label{fig:power_law}
\end{figure}

\subsection{High-order PF interference and Adiabatic evolution}
Here in this section, we show the existence of high-order PF interference. 
We consider the 1D transverse-field Ising (TFI) model with the periodic boundary condition
\begin{equation}
    H_{\tfi} 
    = J \sum_{j=1}^{n} X_jX_{j+1} + h\sum_{j=1}^{n}  Z_j,
\end{equation}
which can be grouped into two groups i.e., 
$H_x=J \sum_{j=1}^{n} X_jX_{j+1}$ and $H_z=h\sum_{j=1}^{n}  Z_j$
where each term in the group commutes with each other.

Generally, the second-order product formula (PF2) does not interfere. Thus, the bound by the triangle inequality is tight for most cases of PF2.  As a special case of \cref{thm:approx-interference-bound}, the error in two-term PF2 interferes when $H_1$ is much larger than $H_2$. The leading term of the error of PF2 is $R=\frac{1}{24}[H_1-H_2, [H_1, H_2]] = \frac{1}{24}[H, [H_1, H_2]] - \frac{1}{12}[H_2, [H_1, H_2]]$, where $H=H_1+H_2$. $[H, [H_1, H_2]]$ is off-diagonal in the eigenbasis of $H$ and satisfies the orthogonality condition, while the norm of $[H_2, [H_1, H_2]]$ is much smaller than $R$ when $H_2$ is much smaller than $H_1$. 
The rigorous statement of this idea can be found in \cref{thm:pf2-two-term} in the methods section. 
In \cref{fig:TFI}, we see that the error of PF2 simulation of traverse Ising model interferes when $H_z$ is significantly smaller than $H_x$.

\begin{figure}[!t]
    \centering
    \includegraphics[width=0.95\linewidth]{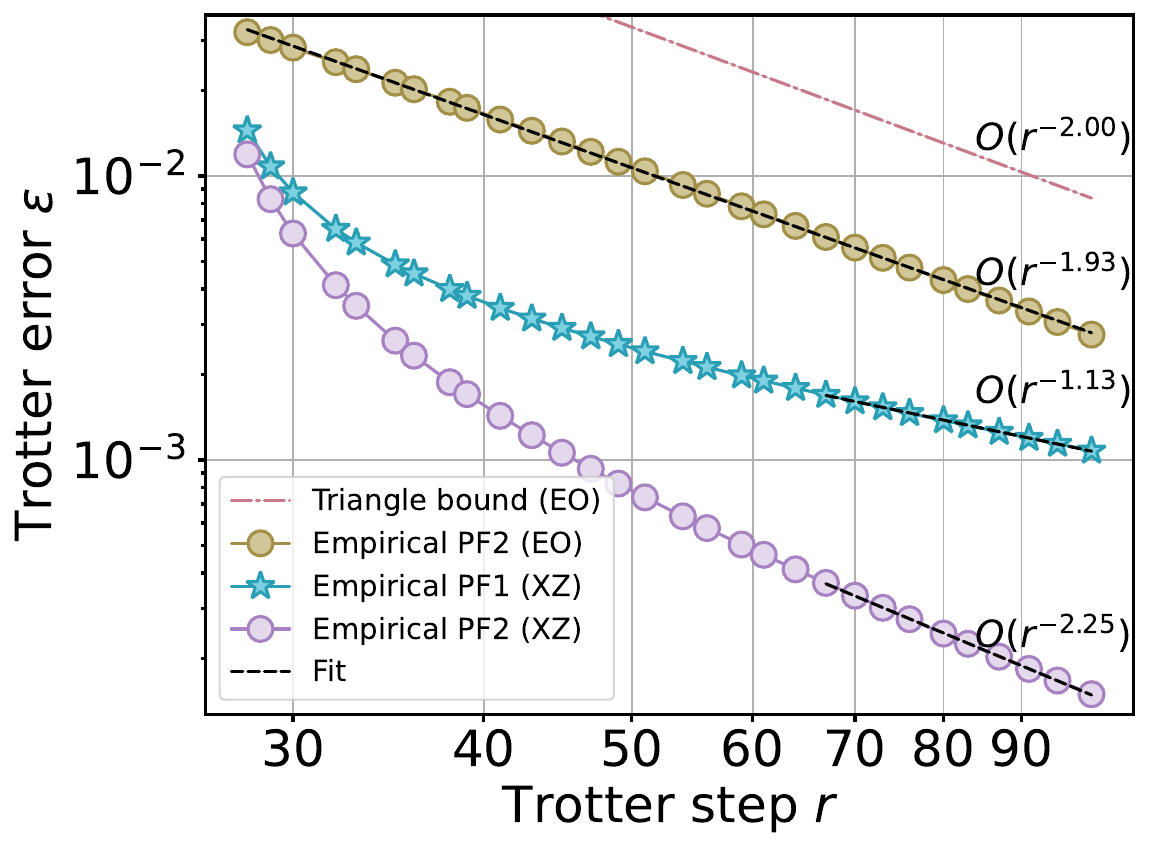}
    \caption{Error interference of both PF1 and PF2 with the XZ-group when the transverse-field Ising (TFI) model with parameters $J=2$ and $h=0.001$.
    The empirical error of PF2 with the even-odd EO-group is plotted for reference.
    }
    \label{fig:TFI}
\end{figure}

Another example is simulating adiabatic evolution where the Hamiltonian consists of two terms, $H[u(\tau)] = (1-u(\tau))H_1+ u(\tau)H_2$. 
Here $u$ is a continuous function satisfying $u(0)=0, u(t)=1$. In \cite{kovalskySelfhealingTrotterError2023}, they show that the error of PF1 of the adiabatic evolution has a self-healing property. Our results suggest that the error interferes when $\tau$ is close to 0 or $t$, even for higher-order PF. 

The biased Hamiltonian condition is also reasonable when simulating perturbation dynamics. It is a natural setup for the simulation of perturbation dynamics where one of the $H_i$'s is significantly larger than all others. A typical perturbation dynamics consists of two dynamics $H=H_1+H_2$, where $H_1$ is large and time-independent, and $H_2$ is small and possibly time-dependent. Thus, our theory also suggests the presence of error interference in the perturbation model. 

\section{Discussion and Outlook}
In this work, we introduced a novel framework for analyzing Trotter error interference in the Product Formula (PF) quantum simulation. By estimating the algorithmic Trotter error for long-time evolution, we derived both necessary and sufficient conditions for strict error interference. Furthermore, we provided a general theoretical lower bound for algorithmic error accumulation, highlighting scenarios where error interference is not possible. Through various examples, we demonstrated the existence of (approximate) error interference phenomena beyond the bipartite Hamiltonian case.

Note that our method is not only applicable to PF methods. We anticipate that it can also be extended to other quantum simulation methods, such as analog simulation, or for handling physical coherent errors, provided they satisfy the conditions for error interference. 

It will be interesting to investigate how we can utilize the error interference phenomena to speed up the quantum simulation process and propose new algorithms. Our current work has shown that error interference can lead to speed-ups in implementing original algorithms. We believe there is potential to further leverage this property to develop more efficient quantum simulation algorithms. 

In conclusion, our work on error interference in PF quantum simulation not only deepens the theoretical understanding of the Trotter error but also provides a practical tool for improving quantum simulation algorithms. We believe that our findings will inspire future research in this area, bringing us closer to realizing the full potential of quantum computation for simulating quantum dynamics. 

\section{Methods}
\subsection{Second order product formula}
For higher order PF, we can show the error interferes when one of $H_l$ is much larger than all others. 
We show only rigorously for two-terms PF2, but similar techniques can be applied to higher-order PF with more terms. As for PF2 with two terms, $\pf_2(\dt) = e^{-\ii H_1 \dt/2}e^{-\ii H_2\dt}e^{-\ii H_1\dt/2}$, the leading error and the high-order remaining error term can be expressed as
\begin{align*}
    R &=  \frac{1}{24}[H_2, [H_2,H_1]] + \frac{1}{24}[H_1, [H_1,H_2]],\\
    \norm{R_{\re}} &= \bigO \qty(\norm{[H_1, [H_2, [H_1, H_2]]]} + \norm{[H_2, [H_1, [H_1, H_2]]]}).
\end{align*}
Thus the error of PF2 interferes only if $R$ and $H$ satisfy the orthogonal condition, or at least approximately satisfy the orthogonal condition. In the following theorem, we show the case that $R$ and $H$ satisfy the orthogonal condition. When one of the $H_i$ is much greater than the other, the error interferes approximately.

\begin{theorem}[Approximate error interference of PF2 of two terms when one term is the major term]\label{thm:pf2-two-term}
    For $H = H_1 + H_2$ and PF2 formula $\pf_2(\tau) = e^{-\ii H_1 \tau /2} e^{-\ii H_2 \tau} e^{-\ii H_1 \tau/2}$, the Trotter error $\norm{\pf_2(t/r)^r - e^{-\ii Ht}} $ can be bounded as
    \begin{equation}\label{eq:approx_pf2}
        \bigO \qty(\norm{[H_1, H_2]} \frac{t^3}{r^3} + \norm{[H_2, [H_1, H_2]]}\frac{t^3}{r^2} + 
        \cml \frac{t^4}{r^3}),
    \end{equation}
    where $\cml = \norm{[H_1, [H_2, [H_1, H_2]]]} + \norm{[H_2, [H_1, [H_1, H_2]]]}$.
\end{theorem}
The bound in the theorem will improve the triangle bound when $H_1$ is much larger than $H_2$.
The proof of this theorem can be found in~\cref{apd:sec:proof-approximate-pf2-interference}.

\subsection{Exact interference bound}

In \cref{thm:fast-convergence}, we work with large $\bigO$ factors, ignoring the hidden constant in the convergence rate. We need the explicit constant factors when we want to give an explicit bound of the error. So here we will prove theorems with explicit constant factors. The following theorem is a more explicit form of \cref{thm:fast-convergence} in the case of PF1. In the case of PF1, we can show rigorously that the error interferes and grows sublinearly. 
\ifnum\aqis=1
Here we provide the statement for tight bound in the general case. The statement of tight bound for PF1 will be shown in the technical version.
\fi

\begin{theorem}[Tight upper bound for general error interference]\label{thm:tight-general-bound}
    Let $H$ and $R$ be Hermitian, and $h$ be any real parameter. Let $\{\ket{\psi_i}\}$ and $\{\ket{\psi_i'}\}$ to be the eigenvectors of $H$ and $H+hR$ respectively and $\{\lambda_i\}, \{\lambda_i'\}$ to be the corresponding eigenvalues. Then $R$ can be decomposed as $R = \sum b_{jk}\ket{\psi_j}\bra{\psi_k'}$. For any $\epsilon>0$, let 
    \begin{align*}
    \Delta_H^\epsilon(R) &= \max\Big\{\big\|\sum_{|\lambda_j - \lambda_k'| <\epsilon}b_{jk}\ket{\psi_j}\bra{{\psi}_k'}\big\|, \\&\big\|\sum_{|\lambda_j - \lambda_k'| <\epsilon}\sgn(\lambda_j-\lambda_k')b_{jk}\ket{\psi_j}\bra{{\psi}_k'}\big\|\Big\}, \\
    \calR_H^\epsilon(R) &= \sum_{|\lambda_j - \lambda_k'|\geq \epsilon}\frac{1}{\lambda_j - \lambda_k'} b_{jk} \ket{\psi_j}\bra{\psi_k'}
    \end{align*}
    Then we have
    \ifnum\aqis=0
    \begin{equation}\label{eq:tight-interference-general}
        \norm{e^{-\ii(H+hR)t} - e^{-\ii Ht}} \leq (1+t\epsilon)h\max_{ |\epsilon'|\leq \epsilon}\Delta_H^{\epsilon'}(R)t + 2h\norm{\calR_H^\epsilon(R)}.
    \end{equation}
    \else
    \begin{align}
        &\norm{\exp(-\ii(H+hR)t) - \exp(-\ii Ht)} \\ 
        &\leq h(1+t\epsilon)\max_{ |\epsilon'| \leq \epsilon}\norm{\Delta_H^{\epsilon}(R)}t + 2h\norm{\calR_H^\epsilon(R)}.
    \end{align}
    \fi
\end{theorem}
The proof of the theorem can be found in~\cref{apd:sec:proof-exact-interference}.
    The definitions of $\Delta_H^\epsilon(R)$ and $\calR_H^\epsilon(R)$ might seem not intuitive at first glance, but it is just an alternative expression for the diagonal element of $R$ in eigenbasis of $H$. 
    $\Delta_H^\epsilon(R)$ represents the matrix formed by the near-diagonal terms of $R$ in the eigenbasis of $H$, while $\calR_H(R)^\epsilon$ is the matrix satisfying the constraint $[H, \calR_H^\epsilon(R)] = R -\Delta_H^\epsilon(R)$ when $\epsilon=0$, representing the bound on the accumulated error of the off-diagonal terms. As a special case, we show the exact bound for PF1 in the following theorem.
\ifnum\aqis=0
\begin{theorem}[Tight upper bound for PF1]\label{thm:exact-pf1-interference}
     Assume $H_l$ to be Hermitians. 
    Let $\pf_1(\dt) = \prod_{l}^{\rightarrow} e^{-\ii \dt H_l}$ be the first-order product formula for $H=\sum H_l$. Set $\dt = t/r$ and
    $\pf_1(\dt)$ is repeated $r$ times to simulate $e^{-\ii H t}$. 
    Define $R\coloneqq  \sum_{j<k}[H_j, H_k]$. 
    Define $\ket{\psi_i}$ and $\ket{\psi_i'}$ to be the eigenvectors of $H$ and $H+ \frac{\dt}{2\ii} R$, and the corresponding eigenvectors are $\lambda_i$ and $\lambda_i'$. 
    Assume that $R$ can be decomposed in the basis of $\ket{\psi_j}$ and $\ket{\psi_k'}$ as $R = \sum b_{jk} \ket{\psi_j}\bra{\psi_k'}$. 
    For any $\epsilon$, define $\Delta_H^\epsilon(R)$ and $\calR_H^\epsilon(R)$ similarly to the definition in~\cref{thm:tight-general-bound}. 
    Then we have that
    \begin{equation}\label{eq:tight-interference-pf1}
    \begin{aligned}
        \norm{\pf_1(\dt)^r - e^{-\ii Ht}} &\leq \cml_1(H)\frac{t^3}{r^2} + (1+t\epsilon)\max_{|\epsilon'|\leq \epsilon}\bigg\|{\Delta_H^{\epsilon'}(R)}\bigg| \frac{t^2}{r}\\& + 2\norm{\calR_H^\epsilon(R)}\frac{t}{r} + \cml_2(H) \frac{t^4}{r^3},
    \end{aligned}
    \end{equation}
    where
    \begin{align*}
        \cml_1(H) &= \frac{1}{2}\sum_{i}\norm{\sum_{i<l<k}[H_i,[H_l,H_k]]} \\&+ \frac{1}6\sum_{i}\norm{\sum_{l>i}[H_i, [H_i, H_l]]},\\
        \cml_2(H) &= \frac{1}{12}\sum_i\norm{\sum_{i<l<k}[H_i,[H_i,[H_l,H_k]]]}.
    \end{align*}
\end{theorem}
Proof of the theorem is shown in~\cref{apd:sec:proof-exact-interference}.

In \cref{eq:tight-interference-general} and \cref{eq:tight-interference-pf1}, the leading term of the error is controlled by the norm of $\Delta_H^\epsilon(R)$ and $\calR_H(R)$. $\Delta_H^\epsilon(R)$ can be viewed as the (almost) diagonal terms of $R$ in the eigenbasis of $H$, while $\calR_H^\epsilon(R)$ represents the off-diagonal terms of $R$ in the eigenbasis of $H$. When $\Delta_H^0(R) = 0$, which is exactly the condition of \cref{thm:fast-convergence}, the error would behave as PF2 as long as $r$ tends to infinity. But if $\calR_H^\epsilon(R)$ is too large, then the error grows linearly until $r$ is large enough. In the real world, we observe that $\calR(H)$ is generally of reasonable scale so we can say that the existence of interference is equivalent to $\Delta_H^0(R)=0$.

\emph{Remark.} As a concrete example, PF2 that each term has equal norm has a large $\Delta_H^0(R)$, meaning the error does not ``interfere'' and would accumulate linearly according to the triangle inequality.
\fi

\subsection{Examples of none-interferencing PF2 formulas}\label{sec:non-interference-pf2}
This section provides a concrete example of two-term PF2 that does not interfere.
Consider the bi-grouping of 1D nearest-neighbor Heisenberg model, with the partition according to the parity (even-odd) of the site index of each Hamiltonian term 
$H_{\nn} = H_{\even} +  H_{\odd}$, i.e.,
\begin{align}
    H_{\even} &= \sum_{j=1}^{\lfloor{n/2\rfloor}} \qty(X_{2j-1}X_{2j} + Y_{2j-1}Y_{2j} + Z_{2j-1}Z_{2j} +hZ_{2j-1}) \\
    H_{\odd}  &= \sum_{j=1}^{\lceil{n/2\rceil}-1} \qty(X_{2j}X_{2j+1} + Y_{2j}Y_{2j+1} + Z_{2j}Z_{2j+1} +hZ_{2j}),
    \label{eq:heisenberg_parity}
\end{align}
where all the summands in $H_{\even}$ (and $H_{\odd}$) commute with each other. Notice that for the leading term of the error in PF2, i.e. $R=\frac{1}{24}[H_{\even}-H_{\odd},[H_{\even}, H_{\odd}]]$, it does not satisfy the orthogonal condition with $H$. 
We can figure this out by observing that $\Tr (HR) = \Theta(n)$. 
So PF2 of the bi-grouping Heisenberg model does not interfere. 
In \cref{fig:tri_group} (a), the numerical simulation shows that the error of PF2 follows exactly the bound of the triangle inequality (cf. \cite[Eq. (152)]{childsTheoryTrotterError2021}).

\emph{Acknowledgment.}--
The authors thank Zhihan Li for the insightful discussions that sparked the initiation of this project. XY is supported by the National Natural Science Foundation of China Grant (Grant No.~12361161602), NSAF (Grant No.~U2330201), and the Innovation Program for Quantum Science and Technology (Grant No.~2023ZD0300200). 
J.X. is supported by the HKU Postgraduate Scholarship. The code used in this study is available on GitHub \url{https://github.com/Jue-Xu/Interference-Trotter-Error}.

\fi

\bibliographystyle{apsrev4-2}
\bibliography{reference, ref_aps}

\ifnum\noappendix=0
\appendix
\onecolumngrid
\newpage

\begin{center}
{\bf \large Supplementary Material: ``Error interference in quantum simulation"} 
\end{center}

\section{Proof of necessary and sufficient condition for error interference}\label{apd:sec:appendix-interference-bound}

\ifnum\supplematerial=1
We list the definition of interference and the orthogonal condition here as we need these definitions in the supplementary material. These definitions are the same as in the main text.
\begin{definition}[Interference, informal]\label{apd:def:interference}
For a $p$th-order product formula $\pf_p(t/r)$ approximating $e^{-\ii Ht/r }$, if the total error grows sublinearly to the sum of error of each timestep, i.e.,
\begin{equation}
    \norm{\pf_p^r(t/r) - e^{-\ii Ht}} = o\qty( r\norm{\pf_p(t/r) - e^{-\ii  H(t/r)}} ),
\end{equation}
where $\norm{\cdot}$ denotes the spectral norm and 
when $r$ is large, then we say that the Trotter error \emph{interferes}.
\end{definition}
\begin{definition}[Orthogonality condition]\label{apd:def:orthogonality}
Consider a given Hamiltonian evolution $e^{-\ii H\dt}$ and its approximation $e^{-\ii H_{\eff}\dt}$.
If the leading-order term $R$ of $H_{\err}=H_{\eff}-H$ satisfies
\begin{align}\label{apd:eq:orthogonality}
    \braket{\psi_i|R|\psi_i}=0
\end{align}
for all eigenstates $\ket{\psi_i}$ of $H$. Then
we say this approximation $e^{-\ii\delta t H_{\eff}}$ satisfies an orthogonality condition.
\end{definition}
\fi
In the following theorem, we prove the necessary and sufficient condition for error interference.
\begin{theorem}[Neccesary and sufficient condition for error interference]\label{apd:thm:appendix-interference-bound}
    The orthogonality condition is a necessary and sufficient condition for Trotter error interference. 
    A PF$p$ evolution can be expressed as $\pf_p(\dt) = \exp(-\ii H\dt -\ii R \dt^{p+1} -\ii R_{\re}(\dt)\dt^{p+2})$, 
    where $R$ is independent of $\dt$ and $R_{\re}(\dt)$ has bounded norm when $\dt$ varies. 
    If $R$ satisfies the \crtlnameref{def:orthogonality}, then the total error can be bounded by
    \begin{align}\label{apd:eq:trotter-interference-approx}
         \norm{\pf_p(t/r)^r - e^{-\ii Ht}} = \bigO \qty(\norm{R}\frac {t^p}{r^p} + \norm{R_{\re}} \frac{t^{p+2}}{r^{p+1}}).
    \end{align} 
    If $R$ does not satisfy the orthogonality condition, then the error term is
    \begin{align}
        \norm{\pf_p(t/r)^r - e^{-\ii Ht}} =   
        \Omega\left(\max_i \big|\braket{\psi_i|R|\psi_i}\big|\frac{t^{p+1}}{r^p}\right).
    \end{align}
\end{theorem}
\begin{proof}[Proof]
Let $h=(t/r)^p$ and $H_{\eff} = H+hR$. 
Then clearly $\norm{\pf_p(t/r) - e^{-\ii H_{\eff}t/r}} \leq \norm{R_{\re}}\frac{t^{p+2}}{r^{p+2}}$. 
By triangle inequality we have 
\begin{align}
    \norm{\pf_p(t/r)^r - e^{-\ii H_{\eff}t}} \leq \norm{R_{\re}}\frac{t^{p+2}}{r^{p+1}}.
\end{align}

Suppose the $n$-qubit Hermitian $H$ has a spectral decomposition $H = P \Lambda P^*$ where $P \in SU(d)$ with $d=2^n$ and $\Lambda$ is diagonal. 
Diagonal elements of $\Lambda$ are exactly the eigenvalues of $H$ (denote them as $\lambda_i$), and their corresponding eigenstates are exactly the columns of $P$ (denote them as $\ket{\psi_i}$). 

We will estimate $\delta \lambda_i$ in the following proof, thus deriving a bound on the total error. When $h$ is small enough, we can treat it as a differential, namely $H_{\eff} = H+ \dd h R$. Then we can also calculate the differential of the eigenvalues and eigenstates when $h$ varies (namely $\dd \lambda_i$ and $\dd \ket{\psi_i}$). Notice that $\dd \lambda_i$ is a first-order approximation of $\delta \lambda_i$, so by estimating $\dd \lambda_i$ we can bound $\delta \lambda_i$. 

We claim that
\begin{align}
    \dd \lambda_i = \braket{\psi_i | \dd H | \psi_i} = \Tr (\ket{\psi_i}\bra{\psi_i} \dd H).
\end{align}

Consider the differential of the eigenvectors $\ket{\psi_i}$ 
\begin{align}
    H\ket{\psi_i} = \lambda \ket{\psi_i}, (H+\dd H) (\ket{\psi_i} + \dd \ket{\psi_i}) = (\lambda_i + \dd \lambda_i)(\ket{\psi_i} + \dd \ket{\psi_i}),
\end{align}
extract all the first-order terms, we have
\begin{align}
    \dd H \ket{\psi_i} + H \dd \ket{\psi_i} = \dd \lambda_i \ket{\psi_i} + \lambda_i \dd \ket{\psi_i}.
\end{align}
Therefore we can deduce that
\begin{align}
    \braket{\psi_i | \dd H | \psi_i }+ \braket{\psi_i| H\,( \dd |\psi_i}) = \braket{\psi_i | \dd \lambda_i | \psi_i} + \braket{\psi_i | \lambda_i \, (\dd | \psi_i}).
\end{align}
Since $\ket{\psi_i}$ is an eigenstate of $H$, we have
\begin{align}
    \braket{\psi_i | H  \textup{d} (| \psi_i}) = \braket{\psi_i | \lambda_i \textup{d} (|\psi_i}).
\end{align}
Consequently, we have
\begin{align}
    \braket{\psi_i | \dd H | \psi_i} = \braket{\psi_i | \dd \lambda_i | \psi_i} = \dd\lambda_i.
\end{align}
Therefore, we have $\dd \lambda_i = 0$ if and only if $\braket{\psi_i| \dd H |\psi_i} =0$. Let's turn to the analysis of variance of $H$ instead of differentials. Namely, we do not assume $h$ tends to 0 for now.
As for \cref{eq:trotter-interference-approx}, we know that if $\delta H =\delta h R$ measures 0 on all eigenstates $\ket{\psi_i}$, then $\delta \lambda_i = \bigO(h^2\norm{R}^2)$. 
Thus, $\delta \Lambda = \bigO(\norm{R}^2h^2)$
$\delta P$ represents the rotation of the eigenvectors such that the difference $\delta P$ can be bounded by the $O(\norm{\delta H}) = O(th)$ when $R$ is fixed. 
So we have
\begin{equation}
    \begin{aligned}
        \exp(-\ii t H_{\eff}) &\approx \exp(-\ii t(H+hR)) \\&= (P+\delta P) \exp (- \ii t(\Lambda + \delta \Lambda) ) (P+ \delta P)^\dagger \\ 
        &= (P + \bigO(\norm{R}h))( \exp (-\ii t\Lambda) + \bigO(\norm{R}^2h^2t) ) (P + \bigO(\norm{R}h))^\dagger \\
        &= P \exp (-\ii t\Lambda) P^\dagger + \bigO(\norm{R}h + \norm{R}^2h^2 t) \\
        &= \exp (-\ii t H ) + \bigO(\norm{R}h + \norm{R}^2h^2t).
    \end{aligned}
\end{equation}

Thus it is immediately that 
\begin{align}
\norm{\pf_p(t/r)^r -
\exp(-\ii Ht)} = \bigO \left( 
\norm{R}\frac{t^{p}}{p^{r}} + \norm{R_{\re}} \frac{t^{p+2}}{r^{p+1}} \right).
\end{align}

On the other hand, the error is $\Omega(thL(R))$ in the case that $L(R) = \max \big| \braket{\psi_i|R|\psi_i} \big| \neq 0$. In fact, we have
\begin{equation}
    \begin{aligned}
        \exp(-\ii tH_{\eff}) &\approx \exp(-\ii t(H+hR)) 
        \\&= (P+\delta P)\exp(-\ii t(\Lambda+\delta\Lambda))(P+\delta P)^\dagger \\
        &= (P + \bigO(\norm{R}h))(\exp(-\ii t\Lambda) + \Omega(L(R)Th))(P+\bigO(\norm{R}h))^\dagger \\
        &= \exp(-\ii tH) + \Omega(L(R)th).
    \end{aligned}
\end{equation}
Thus we have
\begin{align}
    \norm{\pf_p(t/r)^r - \exp(-\ii H t)} = \Omega\left(L(R)\frac{t^{p+1}}{r^p}\right).
\end{align}

The orthogonality condition can be equivalently expressed as there exists some $M$ such that $[H, M] = R$. In fact, if we write $M = \sum b_{ij}\ket{\psi_i}\bra{\psi_j}$, then $[H, M] = \sum b_{ij}(\lambda_i-\lambda_j)\ket{\psi_i}\bra{\psi_j}$, satisfying the orthogonal condition. And if $R = \sum b_{ij} \ket{\psi_i}\bra{\psi_j}$ satisfying $b_{ij}=0$ for all $\lambda_i = \lambda_j$, then we can set $M = \sum \frac{b_{ij}}{\lambda_i - \lambda_j}\ket{\psi_i}\bra{\psi_j}$, and we can check that $[H,M]=R$. Both $H$ and $R$ can be viewed as elements in the Lie group $\mathfrak{su}(d)$, and $R=[H,M]$ means that $R$ is ``tangent'' to $H$.

Also we will show that $L(R) = 0$ implies $\tr (H^nR)=0$, so $\tr (H^nR)=0$ is a necessary condition of fast convergence. In fact, if $R = [H, M]$, then clearly $\Tr (H^k R) = \Tr (H^k[H,M]) = \Tr (H^{k+1}M) - \Tr (H^{k+1}M) = 0$. 
\end{proof}

\begin{corollary}[Error lower bound]\label{apd:cor:lower-bound}
    If the leading-order error term of $\pf_p(t)$, $R$, satisfies  $\Tr(RH^k) \neq 0$ for some $k \in \mathbb{N}$, then $\norm{\pf_p^r(t/r) - e^{-\ii Ht}} = \Omega(\frac{t^{p+1}}{r^p})$ with a sufficient large  $r$.
\end{corollary}
\begin{proof}
    The proof of the corollary is straightforward from \cref{apd:thm:appendix-interference-bound}.
\end{proof}

\section{Proof of interference of a certain state}\label{apd:sec:pf-interference-state}
The interference of error for the evolution of a fixed state is easier to analyze. In this section, we prove two corollaries on error interference for a fixed state.
\begin{corollary}\label{apd:cor-lower-bound-operator}
  Consider a $p$th-order product formula $\pf_p(\dt) = \exp(-\ii H\dt -\ii R \dt^{p+1} -\ii R_{\re}(\dt)\dt^{p+2})$. If an eigenstate of $H$, $\ket{\psi}$, satisfies $\braket{\psi|R|\psi} = 0$, then the total error of the evolution of $\ket{\psi}$ with $\pf_p(t/r)^r$ is
    \begin{equation}
        \norm{(\pf_p(t/r)^r-e^{-\ii Ht})\ket{\psi}} = \bigO \left( \norm{R}\frac{t^p}{r^p} + \norm{R_{\re}}\frac{t^{p+2}}{r^{p+1}} \right).
    \end{equation}   
\end{corollary}
\begin{proof}
    Similarly to the setting in~\cref{apd:thm:appendix-interference-bound}, let $h=(t/r)^p$ and $H_{\eff} = H + hR$ be the effective dynamics of PF$p$. Consider also the case that $h$ tends to zero, so we can treat it as a differential. Then going through the same proof as~\cref{apd:thm:appendix-interference-bound}, we can find that
    \begin{equation}
        \braket{\psi_i|\dd H|\psi_i} = \dd \lambda_i.
    \end{equation}
    Thus $\dd \lambda_i = 0$ if and only if $\braket{\psi_i|\dd H|\psi_i}=0$. Decompose $H$ as $P\Lambda P^\dagger$, where $P$ is unitary and $\Lambda$ is diagonal, then we can find that
    \begin{equation}
    \begin{aligned}
        \exp(-\ii H_{\eff} t) \ket{\psi} &\approx \exp(-\ii (H+hR)t) \ket{\psi}\\
        &= (P + \delta P)\exp(-\ii t(\Lambda + \delta \Lambda))(P + \delta P)^\dagger \ket{\psi}\\
        &= \exp(-\ii t(\Lambda + \delta \Lambda))\ket{\psi} + \bigO(\delta P)  \\
        &= \exp(-\ii \lambda_i t - \ii \delta \lambda_i t)\ket{\psi} + \bigO(\delta P).
    \end{aligned}
    \end{equation}
    Notice that $\left|\exp(-\ii \lambda_i t - \ii \delta \lambda_i t) -\exp(-\ii \lambda_i t)\right| \leq \delta \lambda_i t = \bigO(h^2t)$, thus we are done.
\end{proof}

\begin{corollary}
     Consider a $p$th-order product formula $\pf_p(\dt) = \exp(-\ii H\dt -\ii R \dt^{p+1} -\ii R_{\re}(\dt)\dt^{p+2})$.  
     Let $\ket{\psi}$ be any state. If for any $n$
    \begin{equation}
        \lim_{T \to \infty} \frac 1T \int_0^T \braket{\psi|e^{\ii Ht} RH^n e^{-\ii Ht}|\psi} \dd t = 0,
      \end{equation}
    then the total error of the evolution of $\ket{\psi}$ under the dynamics $\pf_p(t/r)^r$ is
    \begin{equation}
        \norm{(\pf_p(t/r)^r-e^{-\ii Ht})\ket{\psi}} = \bigO \left( \norm{R}\frac{t^p}{r^p} + \norm{R_{\re}}\frac{t^{p+2}}{r^{p+1}} \right).
    \end{equation}
\end{corollary}
\begin{proof}
    Consider the eigenbasis decomposition of $\ket{\psi}$, $\ket{\psi} = \sum a_\lambda \ket{\psi_\lambda}$, then we can find that for any $t$ and $n$
    \begin{equation}
        e^{-\ii Ht}\ket{\psi}\bra{\psi}e^{\ii Ht} = \sum_{\lambda_1, \lambda_2} e^{-\ii Ht(\lambda_1 - \lambda_2)} a_{\lambda_1}\overline{a_{\lambda_2}} \ket{\psi_{\lambda_1}}\bra{\psi_{\lambda_2}}.
    \end{equation}
    Taking average over $t$, we have
    \begin{equation}
        \lim_{T \to \infty} \frac 1T \int_0^T e^{-\ii Ht}\ket{\psi}\bra{\psi}e^{\ii Ht}\dd t = \sum_\lambda |a_\lambda|^2 \ket{\psi_\lambda}\bra{\psi_\lambda}.
    \end{equation}
    Thus we can find that
    \begin{equation}
    \begin{aligned}
        \lim_{T \to \infty} \frac 1T \int_0^T \braket{\psi|e^{\ii Ht} R H^n e^{-\ii H t}} \dd t &= \lim_{T \to \infty} \tr \int_0^T e^{-\ii Ht}\ket{\psi}\bra{\psi}e^{\ii Ht} H^nR\dd t \\
        &= \sum_\lambda |a_\lambda|^2 \tr \ket{\lambda}\bra{\lambda} H^n R \\
        &= \sum_\lambda |a_\lambda|^2 \lambda^n \braket{\psi_\lambda | R |\psi_\lambda}.
    \end{aligned}
    \end{equation}
    Since all $\lambda$ are distinct, $\sum_\lambda |a_\lambda|^2 \lambda^n \braket{\psi_\lambda | R |\psi_\lambda} = 0$ for all $n$ is equivalent to $\braket{\psi_\lambda | R |\psi_\lambda} = 0$ holds for all $\lambda$ (since the Vandermonde determinant is nonzero), so the error caused by $R$ interferes when the initial state is $\ket{\psi}$.
\end{proof}

\section{Proof of tight upper bound (exact interference)}\label{apd:sec:proof-exact-interference}
In this section, we will first provide a tighter upper bound for error interference for general error that can be expressed in the form $\exp(H+\ii hR)$. Then we will show the exact trotter error formula for PF1.
We need a technical lemma dealing with the trotter error:
\begin{lemma}[Variation-of-parameters formula, Theorem 4.9 of \cite{knapp2007basic}]\label{apd:lem:variation-of-param}
    Let $\mathscr{H}(\tau), \mathscr{R}(\tau)$ be continuous operator-valued functions defined for $\tau \in \R$. Then the first order differential equation
    \begin{equation}
        \frac{\d}{\d t}\mathscr{U}(t) = \mathscr{H}(t)\mathscr{U}(t)+\mathscr{R}(t),\, \mathscr{U}(0) \text{ known},
    \end{equation}
    has a unique solution given by the variation-of-parameters formula
    \begin{equation}
        \mathscr{U}(t) = \exp_{\calT}\qty(\int_0^t\d \tau \mathscr{H}(\tau)) \mathscr{U}(0) + \int_0^t \d \tau_1 \exp_{\calT} \qty( \int_{\tau_1}^t \d \tau_2 \mathscr{H}(\tau_2) ) \mathscr{R}(\tau_1),
    \end{equation}
    where $\exp_\calT$ refers to the path-ordered matrix exponential.
\end{lemma}

\begin{theorem}[Tight upper bound for general error interference]\label{apd:thm:tight-general-bound}
    Let $H$ and $R$ be Hermitian, and $h$ be any real parameter. Let $\{\ket{\psi_i}\}$ and $\{\ket{\psi_i'}\}$ to be the eigenvectors of $H$ and $H+hR$ respectively and $\{\lambda_i\}, \{\lambda_i'\}$ to be the corresponding eigenvalues. Then $R$ can be decomposed as $R = \sum b_{jk}\ket{\psi_j}\bra{\psi_k'}$. For any $\epsilon>0$, let $\Delta_H^\epsilon(R) = \max\left\{\norm{\sum_{|\lambda_j - \lambda_k'| <\epsilon}b_{jk}\ket{\psi_j}\bra{{\psi}_k'}}, \norm{\sum_{|\lambda_j - \lambda_k'| <\epsilon}\sgn(\lambda_j-\lambda_k')b_{jk}\ket{\psi_j}\bra{{\psi}_k'}}\right\}$. Define $\calR_H^\epsilon(R) = \sum_{|\lambda_j - \lambda_k'|\geq \epsilon}\frac{1}{\lambda_j - \lambda_k'} b_{jk} \ket{\psi_j}\bra{\psi_k'}$.  Then we have
    \ifnum\aqis=0
    \begin{equation}\label{apd:eq:tight-interference-general}
        \norm{e^{-\ii(H+hR)t} - e^{-\ii Ht}} \leq (1+t\epsilon)h\max_{ |\epsilon'|\leq \epsilon}\Delta_H^{\epsilon'}(R)t + 2h\norm{\calR_H^\epsilon(R)}.
    \end{equation}
    \else
    \begin{align}
        &\norm{\exp(-\ii(H+hR)t) - \exp(-\ii Ht)} \\ 
        &\leq h(1+t\epsilon)\max_{ |\epsilon'| \leq \epsilon}\norm{\Delta_H^{\epsilon}(R)}t + 2h\norm{\calR_H^\epsilon(R)}.
    \end{align}
    \fi
\end{theorem}
\begin{proof}
    Define function $\calL(t) = e^{-\ii (H+hR)t}$, then we can find that
    \begin{align*}
        \frac{\d }{\d t}\calL(t) &= -\ii (H+hR) e^{-\ii (H+hR)t} \\
        &= -\ii H \calL(t) - \ii R e^{-\ii(H+hR)t}.
    \end{align*}
    Thus according to \cref{apd:lem:variation-of-param}, we can find that
    \begin{align*}
        \calL(t) = e^{-\ii H t} + e^{-\ii H t} \int_0^t e^{\ii H \tau} (-\ii h R) e^{-\ii(H+hR)\tau} \d\tau.
    \end{align*}
    Then it is enough to bound $\norm{\int_0^t e^{\ii H \tau} (-\ii h R) e^{-\ii(H+hR)\tau} \d\tau}$. Notice that
    \begin{align*}
        \int_0^t e^{\ii H \tau} (-\ii h R) e^{-\ii(H+hR)\tau} \d\tau &= \int_0^t \sum_{j,k} e^{\ii(\lambda_j-\lambda_k')\tau} b_{jk}\ket{\psi_j}\bra{\psi_k'} \d \tau \\
        &= \sum_{j, k} \frac{1 - e^{\ii(\lambda_j-\lambda_k')t}}{\lambda_j-\lambda_k'}hb_{jk}\ket{\psi_j}\bra{\psi_k'} \\
        &= \sum_{|\lambda_j-\lambda_k'|< \epsilon} \frac{1 - e^{\ii(\lambda_j-\lambda_k')t}}{\lambda_j-\lambda_k'}hb_{jk}\ket{\psi_j}\bra{\psi_k'} + (-\ii h \calR_H^\epsilon(R)) - e^{\ii Ht} (-\ii h\calR_H^\epsilon(R))e^{-\ii (H+hR) t} \\
        &\leq h(1+t\epsilon) \max_{|\epsilon'|\leq \epsilon} \norm{\Delta_H^{\epsilon'}(R)}t + 2h \norm{\calR_H^\epsilon(R)}.
    \end{align*}
    Where the last line stems from the fact that $\sum_{|\lambda_j-\lambda_k'|< \epsilon} \frac{1 - e^{\ii(\lambda_j-\lambda_k')t}}{\lambda_j-\lambda_k'}hb_{jk}\ket{\psi_j}\bra{\psi_k'} - \ii\Delta_H^\epsilon(R)t$ can be written as linear combination of $\Delta_H^{\epsilon'}(R)$ with $\epsilon'\leq \epsilon$ by treating real and imaginary parts separately and bound the linear combination with triangle inequality. So we are done.
\end{proof}
\begin{theorem}[Tight upper bounds for PF1]\label{apd:thm:tight_upper_pf1}
    Assume $H_l$ to be Hermitians. 
    Let $\pf_1(\dt) = \prod_{l}^{\rightarrow} e^{-\ii \dt H_l}$ be the first-order product formula for $H=\sum_l H_l$. Set $\dt = t/r$ and
    $\pf_1(\dt)$ is repeated $r$ times to simulate $e^{-\ii H t}$. 
    Define $R\coloneqq  \sum_{j<k}[H_j, H_k]$. 
    Define $\ket{\psi_i}$ and $\ket{\psi_i'}$ to be the eigenvectors of $H$ and $H+ \frac{\dt}{2\ii} R$, and the corresponding eigenvectors are $\lambda_i$ and $\lambda_i'$. 
    Assume that $R$ can be decomposed in the basis of $\ket{\psi_j}$ and $\ket{\psi_k'}$ as $R = \sum b_{jk} \ket{\psi_j}\bra{\psi_k'}$. 
    For any $\epsilon$, define $\Delta_H^\epsilon(R)$ and $\calR_H^\epsilon(R)$ similarly to the definition in~\cref{apd:thm:tight-general-bound}. 
    Then we have that
    \begin{equation}\label{apd:eq:tight-interference-pf1}
        \norm{\pf_1(\dt)^r - e^{-\ii Ht}} \leq \cml_1(H)\frac{t^3}{r^2} + (1+t\epsilon)\max_{|\epsilon'|\leq \epsilon}\norm{\Delta_H^{\epsilon'}(R)} \frac{t^2}{r} + 2\norm{\calR_H^\epsilon(R)}\frac{t}{r} + \cml_2(H) \frac{t^4}{r^3},
    \end{equation}
    where
    \begin{align*}
        \cml_1(H) &= \frac{1}{2}\sum_{i}\norm{\sum_{i<l<k}[H_i,[H_l,H_k]]} + \frac{1}6\sum_{i}\norm{\sum_{l>i}[H_i, [H_i, H_l]]},\\
        \cml_2(H) &= \frac{1}{12}\sum_i\norm{\sum_{i<l<k}[H_i,[H_i,[H_l,H_k]]]}.
    \end{align*}
\end{theorem}
\begin{proof}
    First we will deal with the case $\pf_1(\dt) = e^{-\ii \dt A} e^{-\ii \dt B}$. Then we can find that
    \begin{align*}
        \frac{\dd}{\dd t} \pf_1(t) &= -\ii t A e^{-\ii\dt A}e^{-\ii t B} + e^{-\ii\dt A}(-\ii t B)e^{-\ii t B} \\
        &= (-\ii A + -\ii B - t[A, B])\pf_1(t) +(e^{-\ii A t}(-\ii B) - (-\ii B)e^{-\ii Bt})e^{-\ii Bt} + t[A, B]e^{-\ii At}e^{-\ii Bt}.
    \end{align*}
    Notice that according to the variation-of-parameters formula, we have
    \begin{align*}
        e^{-\ii A t}(-\ii B) - (-\ii B)e^{-\ii A t} &= - \int_0^t \dd \tau e^{-\ii A \tau_1} [A, B] e^{-\ii A (t-\tau_1)}.
    \end{align*}
    Multiply the equation with $e^{-\ii B t}$ on the right side, and add it with $t[A,B]e^{-\ii At}e^{-\ii Bt}$, we have
    \begin{align*}
        (e^{-\ii A t}(-\ii B) - (-\ii B)e^{-\ii At})e^{-\ii Bt} + t[A, B]e^{-\ii At}e^{-\ii Bt} &= \int_0^t \dd \tau_1 \qty(e^{-\ii A \tau_1 } (-[A, B]) e^{\ii A \tau_1} + [A, B]) e^{-\ii At}e^{-\ii Bt} \\
        &= \int_0^t \dd \tau_1 \int_0^{\tau_1}  \dd \tau_2 e^{-\ii A\tau_2}\qty(\ii[A, [A, B]])e^{\ii A\tau_2} e^{-\ii At} e^{-\ii Bt}.
    \end{align*}
    Apply variation-of-parameters formula again, we have
    \begin{align*}
        \pf_1(t) = e^{-\ii (A+B)t - [A,B]t^2/2} + \int_0^t \dd \tau_1 e^{-\ii (A+B)\tau_1 - [A,B]\tau_1^2/2} \int_0^{\tau_1} \dd \tau_2 \int_0^{\tau_2} \dd \tau_2 e^{-\ii A \tau_3} (\ii [A, [A, B]]) e^{-\ii (\tau_1-\tau_3)A} e^{-\ii \tau_1 B}.
    \end{align*}
    Thus by triangle inequality, we have
    \begin{align*}
        \norm{\pf_1(t) - e^{-\ii(A+B)t - [A,B]t^2/2}} \leq \frac{1}{6} \norm{[A,[A,B]]}.
    \end{align*}
    Then by induction on number of terms, we can show that for general $\pf_1(t) = \prod_l^\rightarrow e^{-\ii H_l t}$, we have
    \begin{align*}
        \norm{\pf_1(t) - \exp(-\ii H t + R t^2/2)} \leq \norm{\cml_1(H)} t^3 + \norm{\cml_2(H)} t^4.
    \end{align*}
    In fact, if we set $\calE_{n-1} = \frac{t^3}{2}\sum_{i\geq 2}\norm{\sum_{i<l<k}[H_i,[H_l,H_k]]} + \frac{t^3}6\sum_{i\geq 2}\norm{\sum_{l>i}[H_i, [H_i, H_l]]} + \frac{t^4}{12}\sum_{i \geq 2}\norm{\sum_{i<l<k}[H_i,[H_i,[H_l,H_k]]]}$, then we have
    \begin{align*}
        &\norm{\prod_l^\rightarrow e^{-\ii H_l t} - e^{-\ii \sum_l H_l t + Rt^2/2}} \\&\leq \norm{\prod_l^\rightarrow e^{-\ii H_l t} - e^{-\ii H_1 t}e^{-\ii \sum_{l\geq 2}H_l t - \sum_{2\leq l<k}[H_l, H_k]}} + \calE_{n-1} \\
        &\leq \norm{e^{-\ii H_1 t -\ii \sum_{l \geq 2}H_l t - \sum_{2\leq l<k}[H_l, H_k]t^2/2 + [-\ii H_1 t, -\ii \sum_{l \geq 2}H_l t - \sum_{2\leq l<k}[H_l, H_k]t^2/2]} - e^{-\ii \sum_l H_l t + Rt^2/2}} \\&+ \frac{1}{6} \norm{[H_1, [H_1, \sum_{l \geq 2} H_l - \ii \sum_{2\leq l<k}[H_l, H_k]]]} + \calE_{n-1} \\
        &\leq \norm{[-\ii H_1 t, - \sum_{2\leq l<k}[H_l, H_k]t^2/2]} + \frac{t^3}{6} \norm{[H_1, [H_1, \sum_{l \geq 2} H_l]]} + \frac{t^4}{12} \norm{[H_1,[H_1, \sum_{2\leq l < k}[H_l, H_k]]]} + \calE_{n-1} \\
        &\leq \frac{t^3}{2} \norm{\sum_{2\leq l<k}[H_1,[H_l, H_k]]} + \frac{t^3}{6}\norm{[H_1,[H_1, \sum_{l\geq 2} H_l]]} + \frac{t^4}{12} \norm{[H_1,[H_1, \sum_{2\leq l < k}[H_l, H_k]]]} + \calE_{n-1} \\
        &= \frac{t^3}{2}\sum_{i}\norm{\sum_{i<l<k}[H_i,[H_l,H_k]]} + \frac{t^3}6\sum_i\norm{\sum_{l>i}[H_i, [H_i, H_l]]} + \frac{t^4}{12}\sum_i\norm{\sum_{i<l<k}[H_i,[H_i,[H_l,H_k]]]}.
    \end{align*}
    Then the bound follows immediately from \cref{apd:thm:tight-general-bound}.
\end{proof}

\section{Proof of approximate interference of Trotter errors}\label{apd:sec:proof_approximate_interference}
In this section, we present the proof of theorems on approximate interference.
\begin{lemma}\label{apd:thm:approx-interference-bound}
    Consider a $p$th-order product formula $\pf_p(\dt) = \exp(-\ii H\dt -\ii R \dt^{p+1} -\ii R_{\re}(\dt)\dt^{p+2})$. 
    If the leading-order error term $R=R_1+R_2$, and $R_1$ satisfies the \crtlnameref{def:orthogonality}, then the total error $\norm{\pf_p(t/r)^r - e^{-\ii Ht}}$ can be bounded by
    \begin{align}
        \bigO\left( \norm{R_2}\frac{t^{p+1}}{r^p} +
        \norm{R_1}\frac{t^p}{r^p} + \norm{R_{\re}} \frac{t^{p+2}}{r^{p+1}} \right).
    \end{align}  
\end{lemma}
\begin{proof}
    Let $h=\frac{t^p}{r^p}$. Take derivation of $\exp(-\ii t(H+hR))$ over $t$, we have
    \begin{equation}
        \frac{\d}{\d t} \exp(-\ii t(H+hR)) = -\ii( H+hR )= -\ii(H+hR_1) - \ii hR_2.
    \end{equation}
    According to \cref{apd:lem:variation-of-param}, we have
    \begin{equation}\label{apd:eq:error-var-of-param}
        \exp(-\ii t(H+hR)) = \exp(-\ii t(H+hR_1)) + \int_0^t \d \tau \exp\qty( -i(t-\tau)H )(-\ii hR_2).
    \end{equation}
    Thus
    \begin{equation}\label{apd:eq:R2-error-bound}
        \norm{\exp(-\ii t(H+hR)) - \exp(-\ii t(H+hR_1))} \leq th \norm{R_2}.
    \end{equation}
    So we are done.
\end{proof}

\section{Proof of approximate error interference for second-order Trotter formula}\label{apd:sec:proof-approximate-pf2-interference}
In this section, we present the theorem of approximate error interference for two-term PF2 when one term is the major term.
\begin{theorem}[Approximate error interference of PF2 of two terms when one term is the major term]\label{apd:thm:pf2-two-term}
      For $H = H_1 + H_2$ and PF2 formula $\pf_2(\tau) = e^{-\ii H_1 \tau /2} e^{-\ii H_2 \tau} e^{-\ii H_1 \tau/2}$, the Trotter error $\norm{\pf_2(t/r)^r - e^{-\ii Ht}} $ can be bounded as
    \begin{equation}\label{apd:eq:approx_pf2}
        \bigO \qty(\norm{[H_1, H_2]} \frac{t^3}{r^3} + \norm{[H_2, [H_1, H_2]]}\frac{t^3}{r^2} + 
        \cml \frac{t^4}{r^3}),
    \end{equation}
    where $\cml = \norm{[H_1, [H_2, [H_1, H_2]]]} + \norm{[H_2, [H_1, [H_1, H_2]]]}$.
\end{theorem}
\begin{proof}
    According to the approximation of PF2 from BCH formula (cf., for example \cite{childsTheoryTrotterError2021}), we get
    \begin{align*}
        \norm{\pf_2(\dt) - \qty(e^{-\ii H \dt} + \frac{\ii}{24}[H_2, [H_2, H_1]] + \frac{\ii}{24}[H_1, [H_1, H_2]])} = \bigO \qty(\cml \dt^4).
    \end{align*}
    Notice that the leading term of error of $\pf_2$ can be written is the form of $[H_1+H_2, \calR]$ when either $H_1$ is much larger than $H_2$
    \begin{align*}
        \frac{1}{24}[H_2,[H_2,H_1]] + \frac{1}{24}[H_1, [H_1, H_2]] = \frac{1}{24}[H_1+H_2, [H_1, H_2]] - \frac{1}{12} [H_2, [H_1, H_2]].
    \end{align*}
    So we can approximate the error as
    \begin{align*}
        \norm{\pf(t/r)^r - \qty(e^{-\ii H t/r} + \frac{\ii}{24}[H_1+H_2, [H_1, H_2]])} \leq \frac{1}{12}\norm{[H_2, [H_1, H_2]]}\frac{t^3}{r^2} + \bigO\qty(\cml\frac{t^4}{r^3}).
    \end{align*}
    Notice that in the case when the error can be written in the form of commutator $R = [H_1+H_2, [H_1, H_2]]$, then we will have $\Delta_H^0(R) = 0$ and $\calR_H^0(R) = [H_1, H_2]$ where $\Delta_H^0(R)$ and $\calR_H^0(R)$ are defined in \cref{apd:thm:tight-general-bound}. Thus according to \cref{apd:thm:tight-general-bound}, we have
    \begin{align*}
        \norm{e^{-\ii(H - 1/24[H, [H_1, H_2]] (t/r)^2)t} - e^{-\ii Ht}} \leq \bigO\qty(\frac{t^3}{r^2} \norm{[H_1, H_2]}).
    \end{align*}
    So we are done by combining all the inequalities.
\end{proof}

\section{Numerical illustration of interference in Heisenberg model}
In this section, we will show some evidence that the error interferes with the first-order product formula simulation.

Recall the dynamics of 1D Heisenberg Hamiltonian
\begin{equation}
\begin{aligned}
    H_{\nn} &= H_X + H_Y + H_Z \\
    &= J_x \sum_{j=1}^{n-1}X_jX_{j+1} + J_y \sum_{j=1}^{n-1} Y_jY_{j+1} + \sum_{j=1}^{n-1} (J_z Z_jZ_{j+1} + hZ_j),
\end{aligned}
\end{equation}
can be implemented by PF1 with tri-group (XYZ)
\begin{align}
    \pf_1 (\dt) = e^{-\ii \dt H_X} e^{-\ii \dt H_Y} e^{-\ii \dt H_Z}.
\end{align}
Then the leading term of error is
\begin{equation}
\begin{aligned}
     R &= \frac 12 \qty([H_X, H_Y] + [H_X, H_Z] + [H_Y, H_Z]) \\
     &= \frac 12 [H_X+H_Y, H_{\nn}] + \frac 12 [H_X, H_Y].
\end{aligned}    
\end{equation}
Notice that the first part of $R$ is skew-symmetric in the eigenbasis of $H_{\nn}$ since it is a commutator of $H_{\nn}$ with some other matrices. The second part of $R$ can be expressed as
\begin{equation}
\begin{aligned}
    [H_X, H_Y] &= J_xJ_y \sum_j (X_jZ_{j+1}Y_{j+2} - Y_jZ_{j+1}X_{j+2}). 
\end{aligned}
\end{equation}
We claim that $\Tr [H_X, H_Y]H_{\nn}^r = 0$ for any $r \geq 1$, so $[H_X, H_Y]$ vanishes on the diagonal of eigenbasis of $H_{\nn}$ except for the degenerated ones.
\fi

\end{document}
